\documentclass[10pt]{article}

\usepackage{amsmath}
\usepackage{amsfonts}
\usepackage{amsthm}
\usepackage[english]{babel}

\newtheorem{theorem}{Theorem}
\newtheorem{lemma}{Lemma}
\newtheorem{definition}{Definition}
\newtheorem{corollary}{Corollary}
\newtheorem{proposition}{Proposition}
\newtheorem{remark}{Remark}

\newcommand{\mR}{\mathbb{R}}
\newcommand{\mC}{\mathbb{C}}
\newcommand{\mN}{\mathbb{N}}
\newcommand{\mE}{\mathbb{E}}
\newcommand{\mZ}{\mathbb{Z}}
\newcommand{\mS}{\mathbb{S}}

\newcommand{\cM}{\mathcal{M}}

\newcommand{\cP}{\mathcal{P}}
\newcommand{\cC}{\mathcal{C}}

\newcommand{\ux}{\underline{x}}
\newcommand{\uxb}{\underline{x} \grave{}}
\newcommand{\uxi}{\underline{\xi}}

\newcommand{\pj}{\partial_{x_j}}
\newcommand{\pjb}{\partial_{{x \grave{}}_{j}}}

\newcommand{\pI}{\partial_{x_i}}

\newcommand{\px}{\partial_x}

\newcommand{\upx}{\partial_{\underline{x}}}
\newcommand{\upxb}{\partial_{\underline{{x \grave{}}} }}

\begin{document}
\title{Hermite and Gegenbauer polynomials in superspace using Clifford analysis}

\author{H.\ De Bie\thanks{Research assistant supported by the Fund for Scientific Research Flanders (F.W.O.-Vlaanderen), E-mail: {\tt Hendrik.DeBie@UGent.be}} \and F.\ Sommen\thanks{E-mail: {\tt fs@cage.ugent.be}}
}

\date{\small{Clifford Research Group -- Department of Mathematical Analysis}\\
\small{Faculty of Engineering -- Ghent University\\ Galglaan 2, 9000 Gent,
Belgium}}

\maketitle

\begin{abstract}
The Clifford-Hermite and the Clifford-Gegenbauer polynomials of standard Clifford analysis are generalized to the new framework of Clifford analysis in superspace in a merely symbolic way. This means that one does not a priori need an integration theory in superspace. 
Furthermore a lot of basic properties, such as orthogonality relations, differential equations and recursion formulae are proven.
Finally, an interesting physical application of the super Clifford-Hermite polynomials is discussed, thus giving an interpretation to the super-dimension.
\end{abstract}

\textbf{MSC 2000 :}   30G35, 58C50, 42C05\\
\textbf{PACS 2006 :}  02.30.Fn, 02.30.Gp\\
\noindent
\textbf{Keywords :}   Clifford analysis, Hermite polynomials, Gegenbauer polynomials, superspace, Dirac operator

\newpage
\section{Introduction}

Clifford analysis offers a natural generalization of the theory of complex holomorphic functions in the plane to higher dimension, in which the Dirac operator $\upx$ generalizes the Cauchy-Riemann operator $\partial_z = \partial_x + i\partial_y$. Basic references for this mathematics field are \cite{MR697564,MR1169463,MR1130821}.

In some recent papers (\cite{DBS1,DBS4,DBS2,DBS5}) we have started to develop an extension of Clifford analysis to superspace, i.e. a space with not only commuting but also anti-commuting variables. These superspaces are of great importance in modern theoretical physics. The basic framework, containing all necessary symbols and operators such as Dirac, Euler and Gamma operators was developed in \cite{DBS1} and \cite{DBS4}. Next in \cite{DBS2,DBS5} we studied spherical monogenics, which are polynomial null-solutions of the super Dirac operator. We were able to construct a.o. bases for spaces of spherical monogenics and a formal type of integration over the supersphere leading to the Berezin integral (see \cite{MR732126}).

An important tool in Clifford analysis are orthogonal polynomials such as the Clifford-Hermite and the Clifford-Gegenbauer polynomials. They were introduced in the PhD-thesis of Cnops (\cite{PhDCnops}, see also \cite{MR1169463}) as higher dimensional generalizations of the classical Hermite and Gegenbauer polynomials and have some interesting applications, e.g. in multi-dimensional waveletanalysis (see \cite{MR2106721} and \cite{MR2086651}). Moreover, in \cite{PhDCnops} it is also proven that under certain assumptions these are the only types of classical orthogonal polynomials that can be generalized to the framework of Clifford analysis.

In the present paper these special functions are extended to the superspace setting. This is done by a careful analysis of their definition in the Euclidean case, providing us with a canonical way to generalize them. The main advantage of our approach is that in the resulting special functions one cannot really see the difference between the Euclidean and the superspace case. The only difference is that instead of the Euclidean dimension one has to consider the so-called super-dimension (to be introduced in section \ref{superspaceframework}) in all formulae.

Note that also other approaches are possible in the construction of special polynomials in superspace. We refer the reader to a.o. \cite{MR2025382}.

The paper is organized as follows: first a brief introduction to Clifford analysis on superspace is given (section 2). Then the Clifford-Hermite polynomials (section 3) and the Clifford-Gegenbauer polynomials (section 4) are defined and some of their basic properties are proven: orthogonality, recurrence relations, differential equation, Rodrigues formula etc. Finally, an interesting physical application and a possible further extension are discussed.

\section{The superspace framework of Clifford analysis}
\label{superspaceframework}

Superspaces are spaces where one considers not only commuting but also anti-commuting co-ordinates (see a.o. \cite{MR732126,MR565567,MR574696,MR1175751}). In our approach to superspace (see \cite{DBS1,DBS4}), we start with the real algebra $\cP = \mbox{Alg}(x_i, e_i; {x \grave{}}_j,{e \grave{}}_j)$, $i=1,\ldots,m$, $j=1,\ldots,2n$
generated by

\begin{itemize}
\item $m$ commuting variables $x_i$ and $m$ orthogonal Clifford generators $e_i$
\item $2n$ anti-commuting variables ${x \grave{}}_i$ and $2n$ symplectic Clifford generators ${e \grave{}}_i$
\end{itemize}
subject to the multiplication relations
\[ \left \{
\begin{array}{l} 
x_i x_j =  x_j x_i\\
{x \grave{}}_i {x \grave{}}_j =  - {x \grave{}}_j {x \grave{}}_i\\
x_i {x \grave{}}_j =  {x \grave{}}_j x_i\\
\end{array} \right .
\quad \mbox{and} \quad
\left \{ \begin{array}{l}
e_j e_k + e_k e_j = -2 \delta_{jk}\\
{e \grave{}}_{2j} {e \grave{}}_{2k} -{e \grave{}}_{2k} {e \grave{}}_{2j}=0\\
{e \grave{}}_{2j-1} {e \grave{}}_{2k-1} -{e \grave{}}_{2k-1} {e \grave{}}_{2j-1}=0\\
{e \grave{}}_{2j-1} {e \grave{}}_{2k} -{e \grave{}}_{2k} {e \grave{}}_{2j-1}=\delta_{jk}\\
e_j {e \grave{}}_{k} +{e \grave{}}_{k} e_j = 0\\
\end{array} \right .
\]
and where moreover all elements $e_i$, ${e \grave{}}_j$ commute with all elements $x_i$, ${x \grave{}}_j$. The algebra generated by all the $e_i$, ${e \grave{}}_j$ is denoted by $\cC$. In the case where $n = 0$ we have that $\cC \cong \mR_{0,m}$, the standard orthogonal Clifford algebra with signature $(-1,\ldots,-1)$. The main anti-involution $\bar{\; . \;}$ on $\mR_{0,m}$ is defined by 

\[
\begin{array}{l}
\overline{a b} = \overline{b} \, \overline{a}, \quad a,b \in \mR_{0,m}\\
\overline{e_j} = - e_j.
\end{array}
\]

The most important element of the algebra $\cP$ is the vector variable $x = \ux+\uxb$ with

\[
\begin{array}{lll}
\ux &=& \sum_{i=1}^m x_i e_i\\
&& \vspace{-2mm}\\
\uxb &=& \sum_{j=1}^{2n} {x \grave{}}_{j} {e \grave{}}_{j}.
\end{array}
\]

The square of $x$ is scalar-valued and equals $x^2 =  \sum_{j=1}^n {x\grave{}}_{2j-1} {x\grave{}}_{2j}  -  \sum_{j=1}^m x_j^2$.

On the other hand the super Dirac operator is defined as

\[
\px = \upxb-\upx = 2 \sum_{j=1}^{n} \left( {e \grave{}}_{2j} \partial_{{x\grave{}}_{2j-1}} - {e \grave{}}_{2j-1} \partial_{{x\grave{}}_{2j}}  \right)-\sum_{j=1}^m e_j \pj.
\]

Its square is the super Laplace operator

\[
\Delta = \px^2 =4 \sum_{j=1}^n \partial_{{x \grave{}}_{2j-1}} \partial_{{x \grave{}}_{2j}} -\sum_{j=1}^{m} \pj^2.
\]

If we let $\px$ act on $x$ we find that

\[
\px x = m-2n = M
\]

where $M$ is the so-called super-dimension. This numerical parameter gives a global characterization of our superspace and will be very important in the sequel.

Furthermore we introduce the super Euler operator
\begin{eqnarray*}
\mE &=& \sum_{j=1}^m x_j \pj+\sum_{j=1}^{2n} {x \grave{}}_{j} \pjb.
\end{eqnarray*}

This operator allow us to decompose $\cP$ as
\begin{eqnarray*}
\cP &=& \bigoplus_{k=0}^{\infty} \cP_k, \quad \cP_k=\left\{ \omega \in \cP \; | \; \mE \omega=k \omega \right\}.
\end{eqnarray*}

Now we have the following

\begin{definition}
An element $F \in \cP$ is a spherical monogenic of degree $k$ if it satisfies
\begin{eqnarray*}
\px F &=&0\\
\mE F &=& kF, \quad \mbox{i.e. $F \in \cP_k$}.
\end{eqnarray*}
Moreover the space of all spherical monogenics of degree $k$ is denoted by $\cM_k$.
\end{definition}

The basic calculational rules for the Dirac operator are given in the following lemma (see \cite{DBS1}).

\begin{lemma}
Let $s \in \mN$ and $R_k \in \cP_k$, then
\begin{eqnarray*}
\px(x^{2s} R_k) &=& 2 s x^{2s-1}R_k + x^{2s} \px R_k\\
\px(x^{2s+1} R_k) &=& (2k + M + 2s) x^{2s}R_k - x^{2s+1} \px R_k.
\end{eqnarray*}
\end{lemma}

We then immediately have that
\begin{corollary}
Let $s \in \mN$ and $P_k \in \cM_k$, then
\begin{eqnarray*}
\px(x^{2s} P_k) &=& 2 s x^{2s-1}P_k\\
\px(x^{2s+1} P_k) &=& (2k + M + 2s) x^{2s}P_k.
\end{eqnarray*}
\label{basicrel}
\end{corollary}

These formulae lead to (see \cite{DBS2})

\begin{theorem}[Fischer decomposition]
Let $M \not \in -2 \mN$. Then $\cP_k$ decomposes as
\begin{equation*}
\cP_k = \bigoplus_{i=0}^k x^i \cM_{k-i}.
\end{equation*}
\end{theorem}

Finally, the dimension of the space $\cM_k$ (i.e. the rank as a free $\cC$-module) can be calculated using a Cauchy-Kowalewskaia extension principle and is given by (see \cite{DBS2})

\[
\dim \, \cM_k = \sum_{i=0}^{\min(k,2n)} \binom{2n}{i} \binom{k-i+m-2}{m-2}.
\]

\begin{remark}
It is also possible to consider larger superalgebras than the algebra $\cP$. This is e.g. necessary if one wants to construct a fundamental solution for the super Dirac operator (see \cite{DBS6}).
\end{remark}

\section{Clifford-Hermite polynomials in superspace}
\label{clhermpol}

In classical Clifford analysis, the Clifford-Hermite polynomials (see \cite{MR1169463}) are defined using the following inner product

\[
(f,g) = \int_{\mR^m} \overline{f (\ux)} g(\ux) e^{\ux^2}dV (\ux)
\]

on $L_2(\mR^m;e^{\ux^2})$, where $\bar{\; . \;}$ is the main anti-involution on the Clifford algebra $\mR_{0,m}$. For our purpose, it suffices to know this inner product for functions of the form $f=\ux^s P_k$, $g=\ux^t P_l$ with $P_k$ and $P_l$ spherical monogenics of degree $k$ respectively $l$ in $\mR^m$. The previous integral can then be rewritten, using spherical co-ordinates $\ux = r \uxi$, as

\begin{eqnarray*}
(\ux^s P_k,\ux^t P_l) &=&\int_{\mR^m}  \overline{P_k} \bar \ux^s \ux^t P_l e^{\ux^2} dV(\ux)\\
&=& \int_0^{\infty} r^k r^s r^t r^l e^{-r^2} r^{m-1}dr \int_{\mS^{m-1}} \overline{P_k(\uxi)}  \underline{\overline{ \xi}}^s \underline{\xi}^t P_l(\uxi) d\Sigma(\uxi)\\
&=&  \frac{1}{2} \Gamma(\frac{k+s+t+l+m}{2}) \int_{\mS^{m-1}} \overline{P_k(\uxi)}  \underline{\overline{ \xi}}^s \underline{\xi}^t P_l(\uxi) d\Sigma(\uxi)\\
\end{eqnarray*}

with $\Gamma(.)$ the Gamma-function. 
Note that this inner product consists of two parts: a radial part and an angular part which is an integration over the unit-sphere. If we consider e.g. the case $s=2a$, $t=2b$ the angular integral simplifies to

\begin{eqnarray*}
\int_{\mS^{m-1}} \overline{P_k(\uxi)}  \underline{\overline{ \xi}}^{2a} \underline{\xi}^{2b} P_l(\uxi) d\Sigma(\uxi) = (-1)^{a+b} \int_{\mS^{m-1}} \overline{P_k(\uxi)}  P_l(\uxi) d\Sigma(\uxi).
\end{eqnarray*}

The remaining integral is an inner product on the space of spherical monogenics and can be left out of our discussion.

So, by introducing the following real vector space in our super-setting

\[
R(P_k) = \left\{ \sum_{j=0}^n a_j x^j P_k  \; \; | \; \; n \in \mN, \, a_j \in \mR \right\}
\]

where $P_k$ is a spherical monogenic of degree $k$, fixed once and for all, one can define a bilinear form on $R(P_k)$. This is done by using the previous calculations, however replacing the Euclidean dimension $m$ by the super-dimension $M$ (see also remark \ref{remarksuperdim}) and leads to the following

\begin{definition}
Let $2 \beta = M+2k$, then the bilinear form $<,>$ on $R(P_k)$ is defined by
\[
\begin{array}{lll}
<x^{2s} P_k,x^{2t} P_k> &=&  (-1)^{s+t} \frac{1}{2} \Gamma(s+t+\beta)\\
<x^{2s+1} P_k,x^{2t} P_k> &=&  0\\
<x^{2s} P_k,x^{2t+1} P_k> &=&  0\\
<x^{2s+1} P_k,x^{2t+1} P_k> &=&  (-1)^{s+t} \frac{1}{2} \Gamma(s+t+\beta +1)\\
\end{array}
\]

extended by linearity to the whole of $R(P_k)$.
\end{definition}

Note that this bilinear form is symmetric, but in general not positive definite (this is only the case if $M \in \mN$, $M>0$). Furthermore it is not defined if and only if $M \in -2 \mN$, due to the singularities of the Gamma-function.

We now introduce the following operator

\[
D_{+} = \px+2x
\]

which satisfies $D_+ (R(P_k)) \subset R(P_k)$ because of Corollary \ref{basicrel}. Now we have the following important property of $<,>$.

\begin{proposition}
The operators $\px$ and $D_{+}$ are dual with respect to $<,>$, i.e.

\[
<D_+ p_i P_k, p_j P_k> = < p_i P_k, \px p_j P_k>,
\]
with $p_i P_k$, $p_j P_k \in R(P_k)$, where $p_i$ and $p_j$ are polynomials in the vector variable $x$.
\label{dualityherm}
\end{proposition}

\begin{proof}
In \cite{MR1169463} the similar proposition in the standard Clifford analysis case is proven by using Stokes's theorem in $\mR^m$. In our case we need a different approach.

We have that

\[
<D_+ x^{2s} P_k,x^{2t} P_k> \; = \; 0 \; = \; <x^{2s} P_k, \px x^{2t} P_k>
\]

and
\begin{eqnarray*}
<D_+ x^{2s+1} P_k,x^{2t} P_k> &=&(2k+2s+M) <x^{2s} P_k,x^{2t} P_k> + 2 <x^{2s+2} P_k,x^{2t} P_k>\\
&=&(2k+2s+M) (-1)^{s+t}  \frac{1}{2} \Gamma(s+t+\beta )\\
&& + 2 (-1)^{s+t+1} \frac{1}{2} \Gamma(s+t+\beta + 1)\\
&=&(-1)^{s+t} \frac{1}{2} \Gamma(s+t+\beta ) \left(2k+2s+M - 2(s+t+\beta)\right)\\
&=&-2 t (-1)^{s+t} \frac{1}{2} \Gamma(s+t+\beta )\\
&=& <x^{2s+1} P_k,2t x^{2t-1} P_k>\\
&=&<x^{2s+1} P_k,\px x^{2t} P_k>.
\end{eqnarray*}

The expression $<D_+ x^{2s} P_k,x^{2t+1} P_k>$ is calculated in the same way.
\end{proof}

Now we arrive at the definition of the Clifford-Hermite polynomials in superspace.
\begin{definition}
Let $P_k$ be a spherical monogenic of degree $k$. Then

\[
H_{t,M}(P_k)(x) = (D_+)^t P_k
\]

is a Clifford-Hermite polynomial of degree $(t,k)$.
\end{definition}

We have that, by Corollary \ref{basicrel}, $H_{t,M}(P_k)(x) = H_{t,M,k}(x)P_k$, where $H_{t,M,k}(x)$ is a polynomial in the vector variable $x$, which does not depend on the specific choice of $P_k$, but only on the integer $k$. So clearly $H_{t,M}(P_k)(x) \in R(P_k)$.

The first few Clifford-Hermite polynomials have the following general form:

\begin{eqnarray*}
H_{0,M}(P_k)(x) &=&P_k\\
H_{1,M}(P_k)(x) &=&2 x P_k\\
H_{2,M}(P_k)(x) &=& [4 x^2 + 2(2k+M)]P_k\\
H_{3,M}(P_k)(x) &=&[8 x^3 + 4(2k+M+2)x]P_k\\
H_{4,M}(P_k)(x) &=&[16 x^4 + 16 (2k+M+2)x^2+4(2k+M+2)(2k+M)]P_k.
\end{eqnarray*}

Now we derive the basic properties of these new polynomials. We first have the following straightforward recursion formula:

\begin{theorem}[Recursion formula]

\[
H_{t,M}(P_k)(x) = D_+ H_{t-1,M}(P_k)(x).
\]
\end{theorem}

The Clifford-Hermite polynomials are orthogonal with respect to $<,>$, as is expressed in the following theorem.

\begin{theorem}[Orthogonality relation]
If $s \neq t$ then
\[
<H_{s,M}(P_k)(x),H_{t,M}(P_k)(x)> = 0.
\]
\end{theorem}

\begin{proof}
Suppose $s>t$. Then 
\begin{eqnarray*}
<H_{s,M}(P_k)(x),H_{t,M}(P_k)(x)>&=& < D_+^s P_k,H_{t,M}(P_k)(x)>\\
&=&  <  P_k, \px^s H_{t,M}(P_k)(x)>\\
&=&0,
\end{eqnarray*}
by Proposition \ref{dualityherm} and Corollary \ref{basicrel}.
\end{proof}

\begin{lemma}
The functions $H_{j,M} (P_k)(x)$, $j= 0,1,2, \ldots$ constitute a basis for $R(P_k)$.
\end{lemma}

\begin{proof}
It suffices to note that the coefficient of $H_{j,M} (P_k)(x)$ in $x^j$ is always different from zero.
\end{proof}

The Clifford-Hermite polynomials are solutions of a partial differential equation in superspace. This equation is given in the following

\begin{theorem}[Differential equation]
$H_{t,M}(P_k)(x)$ is a solution of the following differential equation

\[
\px^2 H_{t,M}(P_k)(x)+ 2 x \px H_{t,M}(P_k)(x) - C(t,M,k) H_{t,M}(P_k)(x) = 0
\]

with 

\[
C(t,M,k) = \left\{ \begin{array}{l} 
2t, \quad \mbox{$t$ even}\\
2(t+M+2k-1), \quad \mbox{$t$ odd.}
\end{array}
\right.
\]
\end{theorem}

\begin{proof}
This theorem can be proven by induction. This is necessary in case $M \in -2\mN$. In the other cases it is also possible to use the method described in \cite{MR1169463}.

We write the following expansion of the Clifford-Hermite polynomials

\[
\begin{array}{lll}
H_{2t,M}(P_k) &=& \sum_{i=0}^t a_{2i}^{2t} x^{2i} P_k\\
H_{2t+1,M}(P_k) &=& \sum_{i=0}^t a_{2i+1}^{2t+1} x^{2i+1} P_k.\\
\end{array}
\]

The recursion formula combined with Corollary \ref{basicrel} leads to the following relation among the coefficients

\[
\begin{array}{lll}
a_{2i}^{2t} &=& (2i+2k+M) a_{2i+1}^{2t-1} + 2 a_{2i-1}^{2t-1} \\
a_{2i+1}^{2t+1} &=& (2i+2)a_{2i+2}^{2t} + 2a_{2i}^{2t}.
\end{array}
\]

We need to prove the following (which one can easily see to be true if $t=0$)

\begin{equation}
\begin{array}{lll}
\px H_{2t,M}(P_k) &=&4t H_{2t-1,M}(P_k)\\
\px H_{2t+1,M}(P_k) &=& 2(2t+2k+M) H_{2t,M}(P_k).\\
\end{array}
\label{pdehermite}
\end{equation}

or, in terms of the $a_j^i$

\[
\begin{array}{lll}
2 i a_{2i}^{2t} &=& 4t a_{2i-1}^{2t-1}\\
(2k+2i+M) a_{2i+1}^{2t+1} &=& 2(2t+2k+M)a_{2i}^{2t}.
\end{array}
\]

Indeed, letting act $D_+$ on (\ref{pdehermite}) then yields the theorem.

Suppose now that formula (\ref{pdehermite}) holds for $H_{t,M}(P_k)(x)$, $t \leq 2s$. We show that it also holds for $t=2s+1$. Indeed

\begin{eqnarray*}
(2k+2i+M) a_{2i+1}^{2s+1} &=& (2k+2i+M) ( (2i+2)a_{2i+2}^{2s} + 2a_{2i}^{2s} )\\
&=&(2k+2i+M) ( 4s a_{2i+1}^{2s-1} + 2a_{2i}^{2s} )\\
&=&4s a_{2i}^{2s} - 8 s a_{2i-1}^{2s-1} + 2(2k+2i+M) a_{2i}^{2s}\\
&=&2(2s+2k+M) a_{2i}^{2s} + 4i a_{2i}^{2s} - 8 s a_{2i-1}^{2s-1}\\
&=& 2(2s+2k+M) a_{2i}^{2s}.
\end{eqnarray*}

Similarly we can prove that if the theorem holds for $t \leq 2s+1$, then it also holds for $t=2s+2$.
\end{proof}

The previous proof can be used to give explicit formulae for the coefficients $a_j^i$ in the expansion of the Hermite polynomials. This yields the following 

\begin{theorem}[Explicit form]
If $M \not \in -2 \mN$, then the coefficients in the expansion of the Clifford-Hermite polynomials take the following form
\begin{eqnarray*}
a_{2i}^{2t} &=& 2^{2t}\left( \begin{array}{l}t\\i \end{array} \right)\frac{\Gamma(t+k+M/2)}{\Gamma(i+k+M/2)}\\
a_{2i+1}^{2t+1} &=& 2^{2t+1}\left( \begin{array}{l}t\\i \end{array} \right)\frac{\Gamma(t+1+k+M/2)}{\Gamma(i+1+k+M/2)}.
\end{eqnarray*}
\label{explformhermite}
\end{theorem}

\begin{proof}
We first prove the formula for $a_{2i}^{2t}$. We have that, using the expressions from the previous proof

\begin{eqnarray*}
a_{2i}^{2t} &=& \frac{2t}{i}a_{2i-1}^{2t-1}\\
&=& \frac{4t(t+k-1+M/2)}{i(i+k-1+M/2)}a_{2i-2}^{2t-2}\\
&=& \ldots\\
&=&2^{2i}\frac{t\ldots (t-i+1)}{i(i-1)\ldots 1}\frac{(t+k-1+M/2)\ldots (t+k-i+M/2)}{(i+k-1+M/2) \ldots (k+M/2)}a_{0}^{2t-2i}\\
&=&2^{2i} \left( \begin{array}{l}t\\i \end{array} \right)\frac{\Gamma(t+k+M/2)\Gamma(k+M/2)}{\Gamma(t+k-i+M/2)\Gamma(i+k+M/2)}a_{0}^{2t-2i}.
\end{eqnarray*}

So we need a formula for $a_{0}^{2t}$. This can be done as follows

\begin{eqnarray*}
a_{0}^{2t} &=& (2k+M) a_{1}^{2t-1}\\
&=&(2k+M)2 \frac{2t+M+2k-2}{2k+M} a_{0}^{2t-2}\\
&=&4 (t+k-1+M/2) a_{0}^{2t-2}\\
&=&\ldots\\
&=& 2^{2t} (t+k-1+M/2)\ldots (k+M/2) a_{0}^{0}\\
&=&2^{2t} \frac{\Gamma(t+k+M/2)}{\Gamma(k+M/2)}.
\end{eqnarray*}

Combining these results gives the desired formula for $a_{2i}^{2t}$. The formula for $a_{2i+1}^{2t+1}$ follows from the observation that

\[
a_{2i+1}^{2t+1} = 2\frac{2t+2k+M}{2i+2k+M} a_{2i}^{2t}.
\]
\end{proof}

Now, using the results on the differential equation of the Clifford-Hermite polynomials, we can obtain a second recursion formula:

\begin{theorem}[Recursion formula bis]

\[
H_{t+1,M}(P_k) = 2x H_{t,M}(P_k)+ C(n,M,k) H_{t-1,M}(P_k).
\]
\end{theorem}

\begin{proof}
\begin{eqnarray*}
H_{t+1,M}(P_k) &=& D_+ H_{t,M}(P_k)\\
&=& (\px +2x) H_{t,M}(P_k)\\
&=& 2x H_{t,M}(P_k) + C(t,M,k) H_{t-1,M}(P_k).
\end{eqnarray*}
\end{proof}

One can also formally introduce a Rodrigues formula in superspace. First we define the generalized Gaussian function

\[
\exp(x^2) = \sum_{k=0}^{\infty} \frac{1}{k!} x^{2k}
\]

which we will manipulate symbolically. We then have the following theorem.

\begin{theorem}[Rodrigues formula] The Clifford-Hermite polynomials take the form

\[
H_{t,M}(P_k)(x) = \exp(-x^2) (\px)^t  \exp(x^2) P_k.
\]
\end{theorem}

\begin{proof}
This follows immediately from the following operator equality on $R(P_k)$:

\[
\exp(-x^2) \px  \exp(x^2) = D_+,
\]
combined with the definition of the Clifford-Hermite polynomials.
\end{proof}

Finally, as $H_{t,M}(P_k)(x) = H_{t,M,k}(x)P_k$ where $H_{t,M,k}(x)$ is a polynomial in the vector variable $x$, it is a natural question to ask whether these polynomials are related to special polynomials on the real line. This is indeed the case. More specifically we have the following

\begin{theorem}
One has that
\begin{eqnarray*}
H_{2t,M,k}(x) &=& 2^{2t} t! L_{t}^{\frac{M}{2} + k-1}(-x^2)\\
H_{2t+1,M,k}(x) &=& 2^{2t+1} t! x L_{t}^{\frac{M}{2} + k}(-x^2),
\end{eqnarray*}
where $L_{n}^{\alpha}$ are the generalized Laguerre polynomials on the real line.
\end{theorem}

\begin{proof}
This follows immediately by comparing the coefficients given in Theorem \ref{explformhermite} with the definition of the generalized Laguerre polynomials:

\[
L_{t}^{\alpha}(x) = \sum_{i=0}^t \frac{\Gamma(t +\alpha +1)}{i! (t-i)! \Gamma(i + \alpha +1)} (-x)^i.
\]
\end{proof}

Let us finally calculate the normalization constants $<H_{t,M}(P_k)(x) , H_{t,M}(P_k)(x) >$.

\begin{theorem}
One has that
\begin{eqnarray*}
<H_{2t,M}(P_k)(x) , H_{2t,M}(P_k)(x) > &=&\frac{1}{2} 4^{2t} t! \Gamma(t + M/2+k)\\
<H_{2t+1,M}(P_k)(x) , H_{2t+1,M}(P_k)(x) > &=&\frac{1}{2} 4^{2t+1} t! \Gamma(t + M/2+k+1).\\
\end{eqnarray*}
\end{theorem}

\begin{proof}
We only do the first one, the other one is similar. We have that
\begin{eqnarray*}
<H_{2t,M}(P_k)(x) , H_{2t,M}(P_k)(x) > &=& \frac{1}{C(2t,M,k)}< D_+ \px H_{2t,M}(P_k)(x) , H_{2t,M}(P_k)(x) >\\
&=& \frac{1}{C(2t,M,k)}< \px H_{2t,M}(P_k)(x) ,\px  H_{2t,M}(P_k)(x) >\\
&=& C(2t,M,k) <  H_{2t-1,M}(P_k)(x) , H_{2t-1,M}(P_k)(x) >\\
&=&\ldots\\
&=& C(2t,M,k) C(2t-1,M,k) \ldots C(1,M,k) <P_k, P_k>\\
&=& C(2t,M,k) C(2t-1,M,k)  \ldots C(1,M,k) \frac{1}{2} \Gamma(\beta).\\
\end{eqnarray*}
Replacing the coefficients $C(i,M,k)$ by their actual values gives the desired formula.
\end{proof}

\begin{remark}
The factor $2$ appearing in our definition of the operator $D_+=\px + 2x$ is a convention. This factor corresponds with the so-called physical definition of the classical Hermite-polynomials on the real line. Moreover if one considers the case where $m=1, n=0$ then clearly $k=0$ and $P_k =1$ as the polynomial null-solutions of the one-dimensional Dirac operator are simply the constants. In this case the Clifford-Hermite polynomials reduce to the classical Hermite polynomials.
\end{remark} 

\begin{remark}
The polynomials $x^t P_k$ and $H_{t,M}(P_k)$ both satisfy the following property:

\[
\px x^t P_k = \left\{ \begin{array}{ll} t x^{t-1} P_k &\quad \mbox{$t$ even} \\ (t-1 + 2k+M)x^{t-1} P_k  &\quad \mbox{$t$ odd}\end{array} \right.
\]

\[
\px H_{t,M}(P_k) = \left\{ \begin{array}{ll} 2t H_{t-1,M}(P_k) &\quad \mbox{$t$ even} \\ 2(t-1 + 2k+M)H_{t-1,M}(P_k)  &\quad \mbox{$t$ odd.}\end{array} \right.
\]

The factor $2$ in the formula for the Hermite polynomials disappears if we use the mathematical definition $D_+ = \px+x$ instead of our physical definition (see also the remark above). So we note that the coefficients are the same in both cases. It is interesting to compare this with the work of Rota et al. in \cite{MR0345826}, where they construct an algebraic theory of special polynomials on the real line. In the terminology of that paper, $x^t P_k$, $t=0,1,2,\ldots$ would be a basic sequence for the operator $\px$ and $H_{t,M}(P_k)$, $t=0,1,2,\ldots$ would be a corresponding Sheffer set.
Moreover, our framework gives a quite natural extension of this theory to higher dimensions and it would be worthwhile to further analyze this correspondence.
\end{remark}

\begin{remark}
The crucial part in our treatment of the Clifford-Hermite polynomials was the replacement of the classical Euclidean dimension $m$ by the super-dimension $M$. The same technique has also been used in \cite{DBS5} to construct an integral on superspace. This integral turned out to be equivalent with the Berezin integral.
\label{remarksuperdim}
\end{remark}

\section{Clifford-Gegenbauer polynomials in superspace}

In the Euclidean case, the Clifford-Gegenbauer polynomials are defined making use of the following inner product on the unit ball $B(1)$ in $\mR^m$ (see \cite{MR1169463})

\[
(f,g)_{\alpha} = \int_{B(1)} \overline{f (\ux)} g(\ux) (1+\ux^2)^{\alpha} dV(\ux).
\]

However, similar to the previous section, a computation of this inner product is sufficient in the case $f=\ux^s P_k$, $g=\ux^t P_l$ with $P_k$ and $P_l$ spherical monogenics of degree $k$ respectively $l$ in $\mR^m$. Under that assumption, the previous integral reduces, using spherical co-ordinates, to

\begin{eqnarray*}
(\ux^s P_k,\ux^t P_l)_{\alpha} &=&\int_{B(1)}  \overline{P_k} \overline{\ux^s} \ux^t P_l (1+\ux^2)^{\alpha} dV(\ux)\\
&=& \int_0^1 r^k r^s r^t r^l (1-r^2)^{\alpha} r^{m-1}dr \int_{\mS^{m-1}} \overline{P_k(\uxi)}  \underline{\overline{ \xi}}^s \underline{\xi}^t P_l(\uxi) d\Sigma(\uxi)\\
&=&  \frac{1}{2} B(\frac{k+s+t+l+m}{2},\alpha+1) \int_{\mS^{m-1}} \overline{P_k(\uxi)}  \underline{\overline{ \xi}}^s \underline{\xi}^t P_l(\uxi) d\Sigma(\uxi)\\
\end{eqnarray*}

with $B(x,y)=\Gamma(x) \Gamma(y)/ \Gamma(x+y)$ the Beta-function.

Again this inner product consists of two parts: a radial part and an angular part which is an inner product on the unit sphere. This second part is treated in the same way as in section \ref{clhermpol}. 
Restricting ourselves to spaces of the type $R(P_k)$ as in section \ref{clhermpol} we are thus lead to the following definition, where we have again replaced the Euclidean dimension $m$ by the super-dimension $M$.

\begin{definition}
Let $2 \beta = M+2k$, then the bilinear form $<,>_{\alpha}$ (parametrized by $\alpha$) is defined by
\[
\begin{array}{lll}
<x^{2s} P_k,x^{2t} P_k>_{\alpha} &=&  (-1)^{s+t} \frac{1}{2} B(s+t+\beta ,\alpha+1)\\
<x^{2s+1} P_k,x^{2t} P_k>_{\alpha} &=&  0\\
<x^{2s} P_k,x^{2t+1} P_k>_{\alpha} &=&  0\\
<x^{2s+1} P_k,x^{2t+1} P_k>_{\alpha} &=&  (-1)^{s+t} \frac{1}{2} B(s+t+\beta +1,\alpha+1)\\
\end{array}
\]

extended by linearity to the whole of $R(P_k)$.
\end{definition}

This bilinear form is well-defined if and only if $\alpha \not \in - \mN$ and $M \not \in -2\mN$.

Now we introduce the following important operator

\[
D_{\alpha} = (1+x^2)\px+2 (1+\alpha)x,
\]

which satisfies $D_{\alpha}(R(P_k)) \subset R(P_k)$ because of Corollary \ref{basicrel}.

This operator behaves well with respect to the bilinear form $<,>_{\alpha}$ as is shown in the following proposition.

\begin{proposition}
The operators $\px$ and $D_{\alpha}$ are dual with respect to $<,>_{\alpha}$, i.e.

\[
<D_{\alpha} p_i P_k, p_j P_k>_{\alpha} = < p_i P_k, \px p_j P_k>_{\alpha+1}, 
\]
with $p_i P_k$, $p_j P_k \in R(P_k)$, where $p_i$ and $p_j$ are polynomials in the vector variable $x$.
\label{dualitygeg}
\end{proposition}

\begin{proof}
It suffices to prove the proposition for $<D_{\alpha} x^{2s+1} P_k,x^{2t} P_k>_{\alpha}$, $<D_{\alpha} x^{2s} P_k,x^{2t+1} P_k>_{\alpha}$, $<D_{\alpha} x^{2s+1} P_k,x^{2t+1} P_k>_{\alpha}$ and $<D_{\alpha} x^{2s} P_k,x^{2t} P_k>_{\alpha}$. We only calculate the first one, the others are completely similar.

\begin{eqnarray*}
&&<D_{\alpha} x^{2s+1} P_k,x^{2t} P_k>_{\alpha}\\
 &=&2(\alpha+1) <x^{2s+2} P_k,x^{2t} P_k>_{\alpha}+(2k+2s+M) <(1+x^2)x^{2s} P_k,x^{2t} P_k>_{\alpha} \\
&=&(-1)^{s+t+1} \frac{1}{2} (2\alpha+2+2k+M+2s) B(s+t+\beta + 1,\alpha+1)\\
&& +(2k+2s+M)(-1)^{s+t} \frac{1}{2}B(s+t+\beta,\alpha+1)\\
&=&(-1)^{s+t} \frac{1}{2}\Gamma(\alpha+1)\left(-(2\alpha+2+2k+M+2s)\frac{\Gamma(s+t+\beta +1)}{\Gamma(s+t+\beta +\alpha+ 2)} \right.\\
&&\left. +(2k+2s+M) \frac{\Gamma(s+t+\beta )}{\Gamma(s+t+\beta +\alpha+1)} \right)\\
&=&(-1)^{s+t} \frac{1}{2}\Gamma(\alpha+1)\frac{\Gamma(s+t+\beta )}{\Gamma(s+t+\beta +\alpha+ 2)} (-(2\alpha+2+2k+M+2s)
(s+t+\beta) \\
&&+(2k+M+2s)(s+t+\beta+\alpha+1))\\
&=&-(-1)^{s+t} \frac{1}{2}\Gamma(\alpha+1)\frac{\Gamma(s+t+\beta )}{\Gamma(s+t+\beta +\alpha+ 2)} (\alpha+1)2t\\
&=&-(-1)^{s+t}\frac{1}{2}B(s+t+\beta ,\alpha+2)2t\\
&=&<x^{2s+1} P_k, \px x^{2t} P_k>_{\alpha+1}.
\end{eqnarray*}
\end{proof}

We are now able to define the Clifford-Gegenbauer polynomials in superspace.

\begin{definition}
Let $P_k$ be a spherical monogenic of degree $k$. Then

\[
C^{\alpha}_{t,M}(P_k)(x) = D_{\alpha} D_{\alpha+1}\ldots D_{\alpha+t-1}  P_k
\]

is a Clifford-Gegenbauer polynomial of degree $(t,k)$.
\end{definition}

Again we have that, by Corollary \ref{basicrel}, $C^{\alpha}_{t,M}(P_k)(x) = C^{\alpha}_{t,M,k}(x)P_k$, where $C^{\alpha}_{t,M,k}(x)$ is a polynomial in the vector variable $x$, which does not depend on $P_k$, but only on the integer $k$.

Explicitly, we find the following form for the first Clifford-Gegenbauer polynomials:

\begin{eqnarray*}
C^{\alpha}_{0,M}(P_k)(x) &=&P_k\\
C^{\alpha}_{1,M}(P_k)(x) &=&2(1+\alpha)x P_k\\
C^{\alpha}_{2,M}(P_k)(x) &=&[2(2+\alpha)(2k+M+2+2\alpha)x^2+2(2+\alpha)(2k+M)]P_k\\
C^{\alpha}_{3,M}(P_k)(x) &=&4(3+\alpha)(2+\alpha)[(2k+M+2\alpha+4)x^3 + (2k+M+2)x]P_k.
\end{eqnarray*}

Now we have the following recursion relation.

\begin{theorem}[Recursion formula]
\[
C^{\alpha}_{t+1,M}(P_k)(x) = D_{\alpha} C^{\alpha+1}_{t,M}(P_k)(x).
\]
\end{theorem}

\begin{proof}
It is immediately calculated that
\begin{eqnarray*}
C^{\alpha}_{t+1,M}(P_k)(x)&=& D_{\alpha} D_{\alpha+1}\ldots D_{\alpha+t}  P_k\\
&=& D_{\alpha} \left( D_{\alpha+1}\ldots D_{\alpha+t}  P_k \right)\\
&=& D_{\alpha} C^{\alpha+1}_{t,M}(P_k)(x).
\end{eqnarray*}
\end{proof}

Clifford-Gegenbauer polynomials of different degree are orthogonal, as is expressed in the following theorem.

\begin{theorem}[Orthogonality relation]
If $s \neq t$ then
\[
<C^{\alpha}_{s,M}(P_k)(x),C^{\alpha}_{t,M}(P_k)(x)>_{\alpha}=0.
\]
\end{theorem}

\begin{proof}
Suppose $s>t$. Then 
\begin{eqnarray*}
<C^{\alpha}_{s,M}(P_k)(x),C^{\alpha}_{t,M}(P_k)(x)>_{\alpha}&=& <D_{\alpha} D_{\alpha+1}\ldots D_{\alpha+s-1}  P_k,C^{\alpha}_{t,M}(P_k)(x)>_{\alpha}\\
&=&  < P_k, (\px)^s C^{\alpha}_{t,M}(P_k)(x)>_{\alpha+s}\\
&=&0,
\end{eqnarray*}
by Proposition \ref{dualitygeg} and Corollary \ref{basicrel}.
\end{proof}

The Clifford-Gegenbauer polynomials also satisfy a partial differential equation in superspace. 

\begin{theorem}[Differential equation]
$C^{\alpha}_{t,M}(P_k)(x)$ is a solution of the following differential equation

\[
(1+x^2)\px^2 C^{\alpha}_{t,M}(P_k)(x)+ 2(\alpha+1) x \px C^{\alpha}_{t,M}(P_k)(x) - C(\alpha,t,M,k) C^{\alpha}_{t,M}(P_k)(x) = 0
\]

with 

\[
C(\alpha,t,M,k) = \left\{ \begin{array}{l} 
(2\alpha +t+1)(t+M+2k-1), \quad \mbox{$t$ odd}\\
t(2\alpha +t+M+2k), \quad \mbox{$t$ even}.
\end{array}
\right.
\]
\end{theorem}

\begin{proof}
The theorem can be proved using induction on $t$. The cases where $t=0,1$ are easily checked.  
We write the following expansion of the Gegenbauer polynomials

\begin{equation}
\begin{array}{lll}
C^{\alpha}_{2t,M}(P_k) &=& \sum_{i=0}^t a_{2i}^{2t,\alpha} x^{2i} P_k\\
C^{\alpha}_{2t+1,M}(P_k) &=& \sum_{i=0}^t a_{2i+1}^{2t+1,\alpha} x^{2i+1} P_k.\\
\end{array}
\end{equation}

The recursion formula combined with Corollary \ref{basicrel} leads to the following relation between the coefficients

\[
\begin{array}{lll}
a_{2i}^{2t,\alpha} &=& (2i+2k+M) a_{2i+1}^{2t-1,\alpha+1} + (2\alpha+2i+M+2k) a_{2i-1}^{2t-1,\alpha+1}\\
a_{2i+1}^{2t+1,\alpha} &=& (2i+2)a_{2i+2}^{2t,\alpha+1} +2(1+\alpha+i)a_{2i}^{2t,\alpha+1}.
\end{array}
\]

We need to prove the following:

\[
\begin{array}{lll}
\px C^{\alpha}_{2t,M}(P_k) &=&2t (2\alpha+2t+M+2k) C^{\alpha+1}_{2t-1,M}(P_k)\\
\px C^{\alpha}_{2t+1,M}(P_k) &=& (2 \alpha+2t+2)(2t+2k+M) C^{\alpha+1}_{2t,M}(P_k)\\
\end{array}
\]

or, in terms of the $a_j^{i,\alpha}$:

\[
\begin{array}{lll}
2 i a_{2i}^{2t,\alpha} &=& 2t (2\alpha+2t+M+2k) a_{2i-1}^{2t-1,\alpha+1}\\
(2k+2i+M) a_{2i+1}^{2t+1,\alpha} &=&(2 \alpha+2t+2)(2t+2k+M)a_{2i}^{2t,\alpha+1}.
\end{array}
\]

Suppose now that the theorem holds for $C^{\alpha}_{t,M}(P_k) $, $t \leq 2s$. We show that it also holds for $t=2s+1$. Indeed,

\begin{eqnarray*}
&&(2k+2i+M) a_{2i+1}^{2s+1,\alpha}\\
&=& (2k+2i+M) ((2i+2)a_{2i+2}^{2s,\alpha+1} +(2+2\alpha+2i)a_{2i}^{2s,\alpha+1})\\
&=&(2k+2i+M) ( 2s (2\alpha+2s+M+2k+2) a_{2i+1}^{2s-1,\alpha+2} +(2+2\alpha+2i)a_{2i}^{2s,\alpha+1})\\
&=&(2k+2i+M)(2+2\alpha+2i)a_{2i}^{2s,\alpha+1} \\
&&+2s (2\alpha+2s+M+2k+2) (a_{2i}^{2s,\alpha+1}  - (2\alpha+2i+M+2k+2) a_{2i-1}^{2s-1,\alpha+2})\\
&=&(2 \alpha+2s+2)(2s+2k+M)a_{2i}^{2s,\alpha+1}\\
&& + (2\alpha+2i+M+2k+2) (2 i a_{2i}^{2s,\alpha+1} - 2s (2\alpha+2s+M+2k+2) a_{2i-1}^{2s-1,\alpha+2})\\
&=&(2 \alpha+2s+2)(2s+2k+M)a_{2i}^{2s,\alpha+1}.
\end{eqnarray*}

Similarly we prove that if the theorem holds for $t \leq 2s+1$, then it also holds for $t=2s+2$.
\end{proof}

Now we can give general formulae for the coefficients of the Clifford-Gegenbauer polynomials, where we use the notations of the previous proof.
\begin{theorem}[Explicit form]
If $M \not \in -2 \mN$ and $\alpha \not \in -\mN$, then the coefficients in the expansion of the Clifford-Gegenbauer polynomials take the following form:
\begin{eqnarray*}
a_{2i}^{2t,\alpha} &=& 2^{2t}\left( \begin{array}{l}t\\i \end{array} \right)\frac{\Gamma(t+k+M/2)}{\Gamma(i+k+M/2)} (\alpha + t + 1)_t (\alpha +t+k+M/2)_i\\
a_{2i+1}^{2t+1,\alpha} &=& 2^{2t+1}\left( \begin{array}{l}t\\i \end{array} \right)\frac{\Gamma(t+1+k+M/2)}{\Gamma(i+1+k+M/2)} (\alpha + t + 1)_{t+1} (\alpha +t+k+M/2+1)_i
\end{eqnarray*}
with $(a)_p = a (a+1) \ldots (a+p-1) $ the Pochhammer symbol.
\label{explformgegenbauer}
\end{theorem}

\begin{proof}
We first prove the formula for $a_{2i}^{2t,\alpha}$. We find that, using the expressions from the previous proof:

\begin{eqnarray*}
a_{2i}^{2t,\alpha} &=& \frac{t}{i} (2\alpha+2t+M+2k) a_{2i-1}^{2t-1,\alpha+1}\\
&=& 4 \frac{t}{i} (\alpha+t+M/2+k)(\alpha+t+1) \frac{t+k+M/2-1}{i+k+M/2-1} a_{2i-2}^{2t-2,\alpha+2}\\
&=& \ldots\\
&=&2^{2i}\left( \begin{array}{l}t\\i \end{array} \right)\frac{(t+k+M/2-1)\ldots (t+k-i+M/2)}{(i+k+M/2-1) \ldots (k+M/2)} (\alpha+t+1)\ldots(\alpha+t+i)\\
&& \times (\alpha+t+M/2+k)\ldots (\alpha+t+M/2+k+i-1)a_{0}^{2t-2i,\alpha+2i}\\
&=&2^{2i} \left( \begin{array}{l}t\\i \end{array} \right)\frac{\Gamma(t+k+M/2)}{\Gamma(t+k-i+M/2)}\frac{\Gamma(k+M/2)}{\Gamma(i+k+M/2)}\\
&& \times (\alpha+t+1)_i (\alpha+t+k+M/2)_i a_{0}^{2t-2i,\alpha+2i}.
\end{eqnarray*}

We need a formula for $a_{0}^{2t,\alpha}$. This can be done as follows:

\begin{eqnarray*}
a_{0}^{2t,\alpha} &=& (2k+M) a_{1}^{2t-1,\alpha+1}\\
&=&(2k+M) \frac{2t+M+2k-2}{2k+M} (2\alpha + 2t+2) a_{0}^{2t-2,\alpha+2}\\
&=&4 (t+k-1+M/2) (\alpha +t+1)a_{0}^{2t-2,\alpha +2}\\
&=&\ldots\\
&=& 2^{2t} (t+k-1+M/2)\ldots (k+M/2)(\alpha +t+1)\ldots (\alpha +2t) a_{0}^{0,\alpha+2t}\\
&=&2^{2t} \frac{\Gamma(t+k+M/2)}{\Gamma(k+M/2)}(\alpha +t+1)_t.
\end{eqnarray*}

Combining these results gives the formula stated in the theorem. The formula for $a_{2i+1}^{2t+1,\alpha}$ follows from

\[
a_{2i+1}^{2t+1,\alpha} = (2\alpha+2t+2 )\frac{2t+2k+M}{2i+2k+M} a_{2i}^{2t,\alpha+1}.
\]
\end{proof}

As we have that $C^{\alpha}_{t,M}(P_k)(x) = C^{\alpha}_{t,M,k}(x)P_k$ with $C^{\alpha}_{t,M,k}(x)$ a polynomial in the vector variable $x$, we can compare this polynomial with special functions on the real line. This leads to the following theorem.

\begin{theorem}
One has that
\begin{eqnarray*}
C^{\alpha}_{2t,M,k}(x) &=& 2^{2t} t! (\alpha + t + 1)_t P_t^{\frac{M}{2} + k-1,\alpha}(1 + 2 x^2)\\
C^{\alpha}_{2t+1,M,k}(x) &=& 2^{2t+1} t! (\alpha + t + 1)_{t+1} x P_t^{\frac{M}{2} + k,\alpha}(1 + 2 x^2),
\end{eqnarray*}
where $P_t^{(\alpha,\beta)}$ are the Jacobi polynomials on the real line.
\end{theorem}

\begin{proof}
This follows immediately by comparing the coefficients given in Theorem \ref{explformgegenbauer} with the definition of the Jacobi polynomials:

\[
P_t^{(\alpha,\beta)} (x) = \frac{\Gamma (\alpha+t+1)}{t!\Gamma (\alpha+\beta+t+1)} \sum_{i=0}^t \left(  \begin{array}{l} t\\ i \end{array} \right) \frac{\Gamma (\alpha + \beta + t + i + 1)}{\Gamma (\alpha + i + 1)} \left(\frac{x-1}{2}\right)^i.
\]
\end{proof}

\section{A physical application: interpretation of the super-dimension}

In basic quantum mechanics, the harmonic oscillator is of the utmost importance. It satisfies the following Schr\"odinger equation in $\mR^m$

\[
\frac{1}{2}\left(\upx^2 - \ux^2 \right) \phi =  E \phi,
\]
where we have used the language of Clifford analysis and units $\hbar=m=\omega=1$.

So a canonical extension of this model to superspace would be
\[
\frac{1}{2} \left(\px^2 - x^2 \right) \phi =  E \phi,
\]

\noindent
where we have replaced the Dirac operator and the vector variable by their super analogues. A direct calculation now shows that every function $\phi$ of the form $\phi = \exp(x^2/2) H_{t,M}(P_k)$ is a solution of this equation with corresponding energy $E = M/2 + (t+k)$. Furthermore, for a given energy $E_T = M/2 + T$, there are exactly 
\begin{eqnarray*}
\sum_{i=0}^T \dim \cM_i &=& \sum_{i=0}^T  \sum_{j=0}^{\min(i,2n)} \binom{2n}{j} \binom{i-j+m-2}{m-2}\\
&=& \sum_{i=0}^{\min(T,2n)} \binom{2n}{i} \binom{T-i+m-1}{m-1}
\end{eqnarray*}

eigenfunctions. The second equality follows from the Fischer decomposition or by a direct calculation (see \cite{DBS2}).

Moreover this is in correspondence with what would physically be expected. Indeed, the number of eigenfunctions with energy $E_T$ is the total number of possibilities for selecting $T$ particles out of a set of $m$ bosonic and $2n$ fermionic particles. 

This can also be seen in the following way. If we put

\[
\begin{array}{llll}
a^+_i = \frac{\sqrt{2}}{2}(x_i - \pI ) &\quad&a^-_i = \frac{\sqrt{2}}{2}(x_i+\pI)\\
\vspace{-1mm}\\
b^+_{2i} = \frac{1}{2}({x \grave{}}_{2i} + 2 \partial_{{x \grave{}}_{2i-1}})&\quad&b^-_{2i} =\frac{1}{2} ({x \grave{}}_{2i-1} + 2 \partial_{{x \grave{}}_{2i}})\\
\vspace{-1mm}\\
b^+_{2i-1} = \frac{1}{2}({x \grave{}}_{2i-1} - 2 \partial_{{x \grave{}}_{2i}})&\quad&b^-_{2i-1} = \frac{1}{2}(-{x \grave{}}_{2i} + 2 \partial_{{x \grave{}}_{2i-1}})\\
\vspace{-1mm}\\
\end{array}
\]

we can rewrite the Hamiltonian as

\[
H= \frac{1}{2}(\px^2-x^2) = \sum_{i=1}^m a_i^+ a_i^- +\sum_{i=1}^{2n} b_i^+ b_i^- + \frac{M}{2}.
\]

As the operators $a_i^{\pm}$, $b_i^{\pm}$ satisfy

\[
\begin{array}{lll}
\left[a_i^{\pm} , a_j^{\pm}\right] = 0&\quad& \{ b_i^{\pm} , b_j^{\pm}\} =0\\
\vspace{-1mm}\\
\left[a_i^{-} , a_j^{+}\right] = \delta_{ij}&\quad& \{ b_i^{+} , b_j^{-}\} =\delta_{ij}\\
\vspace{-1mm}\\
\left[a_i^{\pm} , b_j^{\pm}\right] =0&\quad&\left[a_i^{\mp} , b_j^{\pm}\right] =0,
\end{array}
\]

this is the canonical realization of an oscillator with $m$ bosonic and $2n$ fermionic degrees of freedom. The ground level energy is given by $M/2$, which gives us a physical interpretation of the super-dimension. The reader should compare this approach e.g. with the one given in \cite{MR830398} for the purely fermionic case.

Finally, the Hamiltonian can also be factorized in toto using Clifford numbers. Indeed, putting

\[
\begin{array}{lll}
Q_+ = \frac{1}{2}(\px + x) &\quad& Q_- =\frac{1}{2}(\px - x)
\end{array}
\]

we have that

\[
H = \left\{ Q_+, Q_-\right\}.
\]

\section{Possible generalizations}

It is also interesting to note that, although we have focused on the case of superspace, it is possible to treat these special functions in a more general way. First of all, we need an abstract definition of a type of algebras in which the construction can be done. This leads to the following characterization:

Let $\cP$ be a complex algebra satisfying

\begin{itemize}
\item $\cP$ is graded

\[
\cP = \bigoplus_{k=0}^{\infty} \cP_k \qquad \mbox{with} \qquad  \cP_i \cP_j \subseteq \cP_{i+j}.
\]

\item There is a linear operator $\px$ on $\cP$ with

\[
\px: \cP_k \longrightarrow \cP_{k-1}.
\]

In particular : $\px (\cP_0) = 0$. This operator is called the Dirac operator.

\item There exists an element $x \in \cP_1$ such that

\[
\left\{ x, \px \right\} = 2\mE +M,\qquad M \in \mC
\] 

where $\mE$ is the Euler operator defined as $\mE \cP_k = k \cP_k$ (extended by linearity to the whole of $\cP$). Furthermore we call $M$ the formal dimension of the algebra $\cP$.
\end{itemize}

If $\cP$ is an algebra satisfying the above axioms, one can proceed as in this paper to construct Clifford-Hermite and Clifford-Gegenbauer polynomials in $\cP$. However, it remains an open problem to construct examples of algebras $\cP$ satisfying the previous axioms, and having a formal dimension $M \not \in \mZ$.

\section{Conclusions}
In this paper we have generalized the Clifford-Hermite and the Clifford-Gegenbauer polynomials to the framework of Clifford analysis in superspace. This has been done in a purely symbolic way by exploiting the analogy with the classical case. This approach, where one replaces the Euclidean dimension by the super-dimension, was used in previous work to construct an integral over the supersphere and on the whole superspace, leading in this way to the Berezin integral. 

We have proven the basic properties of these special functions, such as recursion relations, orthogonality, differential equations and the like. We have also established a connection with special functions on the real line. Finally we have shown that the Clifford-Hermite polynomials can be seen as solutions of a super harmonic oscillator, thus proving that our framework of Clifford analysis in superspace is physically relevant and meanwhile giving an interpretation to the super-dimension.

In further work we plan to use the Clifford-Hermite polynomials to study generalizations of the Fourier and Radon transforms to superspace. In particular we will use them to determine a singular value decomposition of the super Radon transform.

\vspace{5mm}

\textit{Acknowledgement}
The authors would like to thank Fred Brackx for a careful proofreading of the manuscript.

\end{document}